\documentclass[1p,number,12pt]{elsarticle}
\usepackage[utf8]{inputenc}
\usepackage[T1]{fontenc}
\usepackage{lmodern}
\usepackage[english]{babel}
\usepackage{amsmath}
\usepackage{amssymb}
\usepackage{graphicx}
\usepackage{wrapfig}
\usepackage{float}
\usepackage{xcolor}

\usepackage{fancyhdr}
\newenvironment{proof}[1][Proof]{\begin{trivlist}
\item[\hskip \labelsep {\bfseries #1}]}{\end{trivlist}}

\newenvironment{remark}[1][Remark]{\begin{trivlist}
\item[\hskip \labelsep {\bfseries #1}]}{\end{trivlist}}
\newenvironment{prop}[1][Proposition]{\begin{trivlist}
\item[\hskip \labelsep {\bfseries #1}]}{\end{trivlist}}

\begin{document}

\begin{frontmatter}
\title{The Multipoint Morisita Index for the Analysis of Spatial Patterns}
\author{Jean GOLAY, Mikhail KANEVSKI, Carmen Delia VEGA OROZCO and Michael LEUENBERGER}
\address{Institute of Earth Surface Dynamics, Faculty of Geosciences and Environment, University of Lausanne, Switzerland. Contact: jean.golay@unil.ch.}

\begin{abstract}
In many fields, the spatial clustering of sampled data points has significant consequences. Therefore, several indices have been proposed to assess the degree of clustering affecting datasets (e.g. the Morisita index, Ripley's $K$-function and Rényi's generalized entropy). The classical Morisita index measures how many times it is more likely to select two sampled points from the same quadrats (the data set is covered by a regular grid of changing size) than it would be in the case of a random distribution generated from a Poisson process. The multipoint version takes into account $m$ points with $m \ge 2$. The present research deals with a new development  of the multipoint Morisita index ($m$-Morisita) for (1) the characterization of environmental monitoring network clustering and for (2) the detection of structures in monitored phenomena. From a theoretical perspective, a connection between the $m$-Morisita index and multifractality has also been found and highlighted on a mathematical multifractal set.  
\end{abstract}

\end{frontmatter}

\section{Introduction}
The spatial clustering of sampled data points is of primary interest in many fields from epidemiology to environmental sciences. Therefore, many indices have been proposed to measure the intensity of such structures. Fundamentally, it is possible to distinguish between:

\begin{itemize}
\item topological measures such as the Voronoi polygons and the Delaunay triangulation \cite{Prepa85,KanMai04,KanPoz09}.
\item statistical measures such as Ripley's $K$-function \cite{Ripl81}, the Morisita index \cite{Mori59,Dal02,Veg2013}, the multipoint Morisita index \cite{Hul90}, Moran's Index \cite{Mora50} and the variance-to-mean ratio \cite{Hul90}.
\item fractal measures such as the box-counting method \cite{lov86,Smith89,TuiaKan08}, the sandbox-counting method \cite{Grass831,Fed88,tel89,TuiaKan08}, the lacunarity index \cite{Mand82,AlCl91,Mand94,Plo96,Smith96,Chen97,Dal02}, the information dimension \cite{Renyi56,Hent83,Seur09} and Rényi's generalized dimensions \cite{Grass832,Hent83,Pal87,Bor93,Gab05,Seur09}.
\end{itemize}  

The present research suggests a new development of the multipoint Morisita index ($m$-Morisita) and demonstrates its connection to multifractality. It then deals with the adaptation of the $m$-Morisita index to (1) the characterization of Environmental Monitoring Networks (EMN) clustering and to (2) the detection of structures in monitored phenomena.\\
  
EMN are composed of measurement sites spatially distributed to assess the intensity of environmental phenomena. In spatial planning, for instance, EMN are essential and often used as decision support tools to reduce death occurrence or to improve the general well-being of societies. Consequently, a good understanding of both the reliability of EMN and the information they provide is of paramount importance and a thorough analysis of EMN data must focus on two fundamental issues:\\

\vspace{-10pt}
\begin{enumerate}
\item When dealing with spatially continuous phenomena, a critical issue is related to the high degree of clustering of many EMN (i.e. measurement sites are distributed in space in a non-homogeneous way). It can indeed lead to regional overestimation or underestimation of risk because of the concept of preferential sampling \cite{TuiaKan08}. In order to deal with such problems and to extract representative information from data, several declustering algorithms exist \cite{KanMai04,GSLIB98}. These algorithms induce a loss of information and must not be performed blindly. That is the reason why it is of crucial importance to characterize (i.e. analyse and quantify) the degree of clustering of any EMN before attempting such operations.\\

\vspace{-10pt}
\item Another issue is related to the detection of structures in monitored phenomena. Traditionally, geostatistical tools, like variography, are used \cite{Chi99,KanMai04}, but variograms are quite sensitive to the multi-scale variability of data, the presence of extremes and outliers and the high clustering of monitoring networks. Some special techniques, like robust variography, extreme removal, non-linear transformations of data and regularized variography, can help but it needs deep expert knowledge and many empirical trials.\\
\end{enumerate}

\vspace{-30pt}
\begin{figure}[H]
\begin{center}
\includegraphics[width=13.5cm]{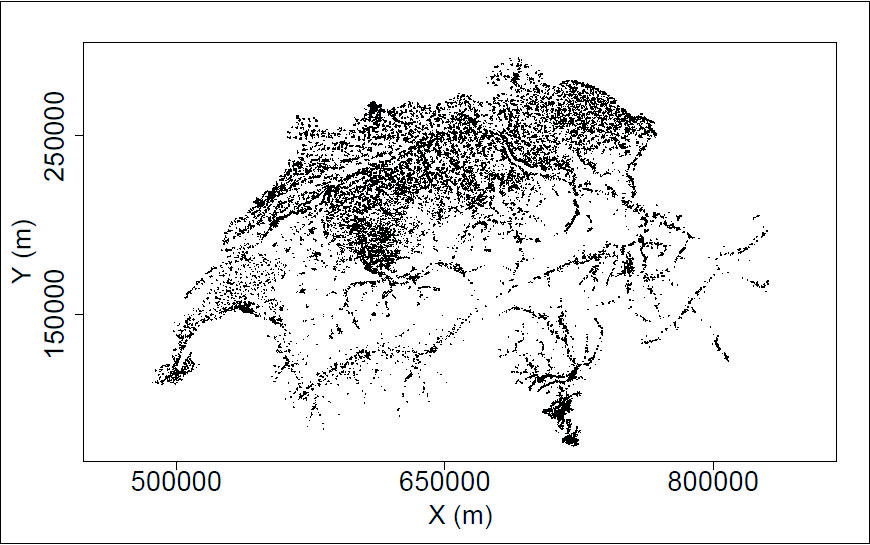}
\end{center}
\vspace{-20pt}
\caption{Postplot of the Swiss Indoor Radon Monitoring Network}
\label{RadonEMN}
\end{figure}
\vspace{-5pt}

In the present research, these critical issues are addressed with the $m$-Morisita index and two methodologies. The first one deals with the first issue and it only takes into account the spatial distribution of EMN measurement sites (i.e. the support of the measures). The second one focuses on the other issue in building a bridge between the support of the measures and the measures themselves through the concept of Functional Measures of Clustering (FMC)\cite{lov87}.\\ 

The $m$-Morisita index is introduced in Section \ref{Mindex} as a generalization of the classical Morisita index. Section \ref{MindexMulti} presents a mathematical connection linking the $m$-Morisita index to multifractality. The content of this section is an interesting contribution of the research since it enables to comprehend the good result provided by the presented index. In Section \ref{cluEMN}, the first methodology based on the $m$-Morisita index for efficiently characterizing EMN is explained and, finally, Section \ref{StrDet} introduces the second methodology for structure detection in monitored phenomena with the $m$-Morisita index. In these last two sections, the challenging case study of the Swiss Indoor Radon Monitoring Network (SIRMN), composed of 57,510 measurement sites, is considered (see Figure \ref{RadonEMN}).

\begin{figure}[H]
\begin{center}
\includegraphics[width=13.5cm]{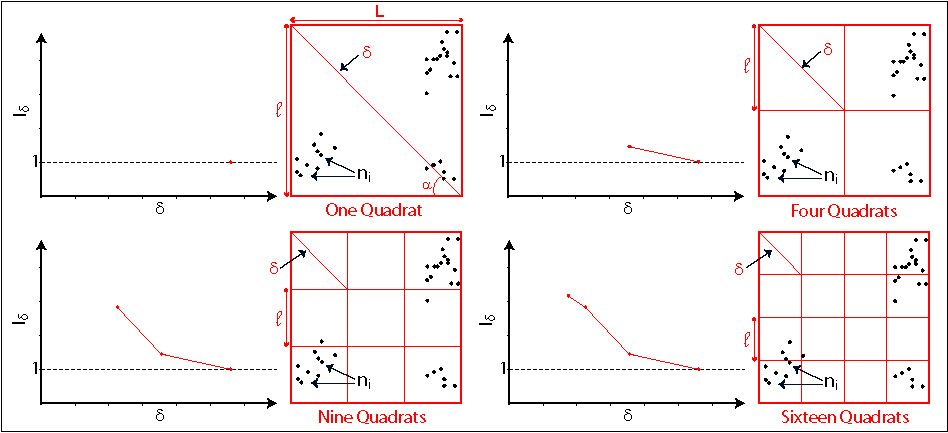}
\end{center}
\vspace{-20pt}
\caption{Computation of the Morisita index for different quadrat sizes}
\label{ComMori}
\end{figure}

\section{The $m$-Morisita index}\label{Mindex}

The classical Morisita index, $I_\delta$, for a chosen quadrat size $\delta$ (i.e. the length of the diagonal), is computed as follows:
\begin{equation}
 I_\delta=Q\hskip5pt\frac{\sum_{i=1}^Q n_i(n_i-1)}{N(N-1)}
\end{equation}
where $Q$ is the number of quadrats necessary to cover the study area, $n_i$ is the number of points in the $i^{th}$ quadrat and $N$ is the total number of points. The shape of the quadrats can be square or rectangular. In two dimensions, $Q$ and $\delta$ are related through $Q = \left( \frac{L}{\cos(\alpha)\ \delta}\right) ^2$ where $L$ and $\alpha$ are respectively the grid length and the angle of the diagonal (see Figure \ref{ComMori}). Properly, the Morisita index measures how many times more likely it is to randomly select two points belonging to the same quadrat than it would be if the points were randomly distributed (i.e. generated from a Poisson process). $I_\delta$ is first calculated for a relatively small quadrat size which is then increased until it reaches a chosen value. It is then possible to draw a plot relating every $I_\delta$ to its matching $\delta$. If the points of the pattern are randomly distributed over the study area, every computed $I_\delta$ fluctuates around the value of 1. If the points are clustered, the number of empty quadrats at small scales increase the value of the index ($I_\delta > 1$) and, finally, if the points are dispersed, the index approaches 0 at small scales \cite{KanMai04}.\\

Notice that, in Figure \ref{ComMori}, quadrats partly overlap from one scale to the next (i.e. the number of quadrats used for the computation of the index throughout the scales does not follow a geometric series). In real case studies (see Sections \ref{StrDet} and \ref{cluEMN}), it is a way of giving more importance to small scales where a change in quadrat sizes is more likely to capture the characteristics of point patterns than great changes at large scales (i.e. it is a kind of regularization). But, from a theoretical perspective, when studying, for instance, mathematical multifractal sets (see Section \ref{MindexMulti}), the number of quadrats of each grid should follow a geometric series with a common ratio $r=2$.\\

Now, the generalization of the classical formulation of the Morisita index, called $m$-Morisita, is made by considering $m$ points with $m \ge 2$ \cite{Hul90}. Strictly speaking, it refers to a family of indices and it is computed from the following formula: \begin{equation}\label{eq:mMorisita}
I_{m,\delta}=Q^{m-1}\frac{\sum_{i=1}^Q n_i(n_i-1)(n_i-2) \dotsm (n_i-m+1)}{N(N-1)(N-2) \dotsm (N-m+1)}
\end{equation}

\begin{figure}[t]
\begin{center}
\includegraphics[width=13.5cm]{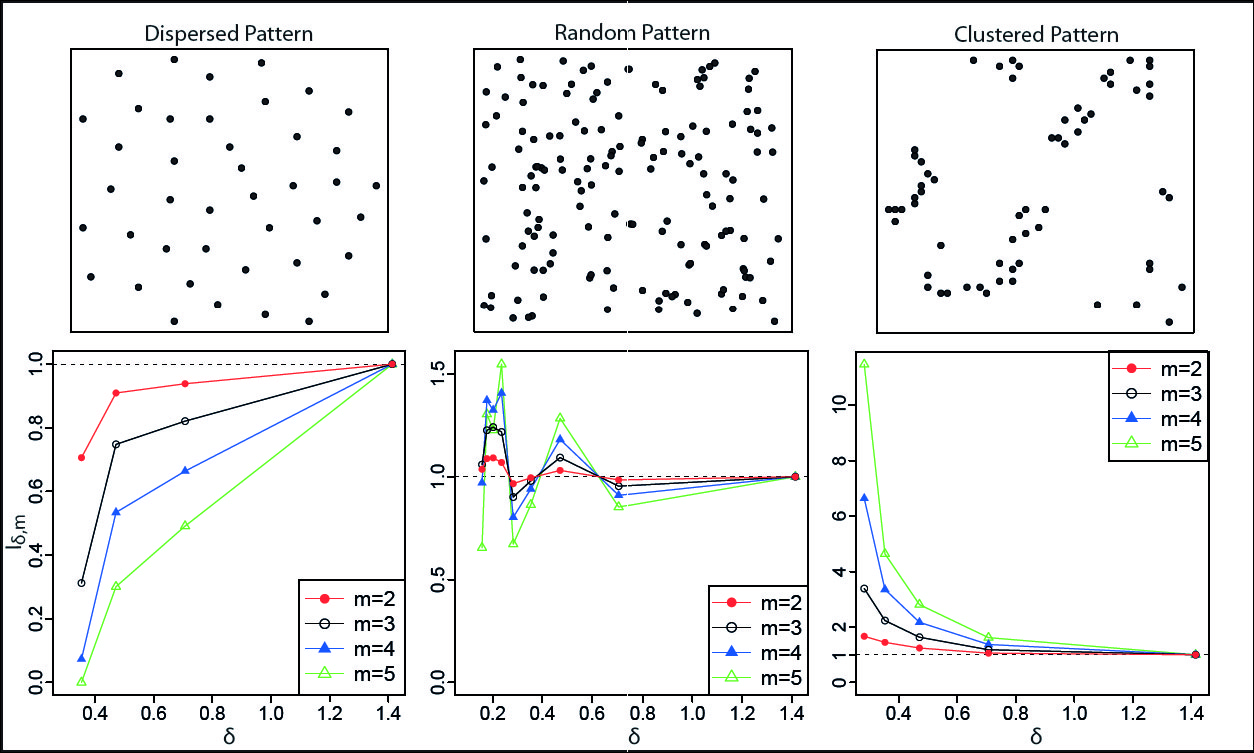}
\end{center}
\vspace{-20pt}
\caption{Results of the $m$-Morisita index for three benchmark distributions and for $m\in\lbrace 2,3,4,5\rbrace$}
\label{TheoMori}
\end{figure}

In the cases of the three benchmark distributions mentioned above (i.e. dispersed, random, clustered), the $m$-Morisita index behaves like the classical one. Nevertheless, as $m$ increases, it becomes more and more sensitive to the structure of the point patterns (see Figure \ref{TheoMori}).\\

Notice that the $m$-Morisita index such as it is defined here is conceptually different from the index thoroughly studied in \cite{Hul90} where the notion of scale is put aside to focus exclusively on the relationship between the values of $m$ and the values of the index for a fixed quadrat size.\\

It is also interesting to mention that, although the classical Morisita index has been applied in a wide range of applications, especially in environmental sciences \cite{ouc86,sha04,bon07,tui07}, the $m$-Morisita index has encountered less success. This lack of interest might have been induced by the difficulties faced in its interpretation. The next section will attempt to solve this problem by presenting the index from a new perspective.

\section{From the $m$-Morisita Index to Multifractality}\label{MindexMulti}

To overcome the interpretation complexities of the $m$-Morisita index, first, it is necessary to understand the connection between the $m$-Morisita index and multifractality. For that, the concepts of fractality and multifractality are introduced by means of the box-counting method, Rényi's generalized dimensions and the lacunarity index. This section is written from the perspective of EMN analysis.

\subsection{Fractality and Clustering} 
\guillemotleft Fractal\guillemotright \ comes from the Latin word \guillemotleft fractus\guillemotright which means \guillemotleft \ irregular\guillemotright \ or \guillemotleft fragmented\guillemotright. It was first coined by \cite{Mand67} to describe sets with abrupt and tortuous edges. A fractal has the two following properties: (1) it is self-similar (at least approximately or stochastically), which means that it reproduces the same structure throughout the scales (i.e. all the moments of the spatial distribution scale in the same way). (2) it has a fractal dimension which is strictly greater than its topological dimension \cite{Mand77, Mand82, Fal03}.\\
 
According to \cite{lov86,Salv97,TuiaKan08}, a fractal dimension can be used to analyse the clustering properties (non-homogeneity) of an EMN and must be interpreted as its dimensional resolution, i.e. the estimation of its ability to detect a phenomenon of a certain dimension in a two-dimensional space. In this context, fractal dimensions range theoretically from 0 (i.e. the topological dimension of a point) to 2 (i.e. the dimension of a geographical space). If the measurement points of an EMN are dispersed or randomly distributed within a rectangular study area, its fractal dimension is close to 2; but this value decreases as the level of clustering increases and it can reach 0 if all the points are superimposed at one location. Thus, fractal dimensions allow us to detect the appearance of clustering as a departure from a dispersed or random situation. 

\subsection{The Box-Counting Method} 
When working with finite data sets, point patterns can be self-similar only over a limited scale interval. The most popular algorithm for estimating their fractal dimension (or Hausdorff-Besicovitch dimension) is the box-counting method (also called the grid method) \cite{lov86, Smith89, TuiaKan08, Seur09}: a regular grid of $Q$ boxes is superimposed on the study area and the number $n_{box}(\delta)$ of boxes necessary to cover the whole dataset is counted; then, the box diagonal of size $\delta$ is reduced and the number $n_{box}(\delta)$ is calculated again. The algorithm goes on till a minimum $\delta$ size is reached. For a fractal point pattern, $n_{box}(\delta)$ follows a power law: 
\begin{equation}
n_{box}(\delta) \propto \delta^{-df_{box}}
\end{equation}
where $df_{box}$ is the fractal dimension measured with the box-counting method. Certainly, in most cases, real point patterns are not mathematical and $df_{box}$ must be interpreted as the dimension of the fractal set most similar to the one under study. It is then possible to consider $-df_{box}$ as the slope of the linear regression fitting the data of the plot relating $\log(n_{box}(\delta))$ to $\log(\delta)$.    

\subsection{Rényi's Generalized Dimensions}
Fractality is purely a geometrical concept and it was made explicit with the $df_{box}$ which only depends on the shape of objects. Nevertheless, in complex cases, fractal sets cannot be fully characterized by only one fractal dimension. In the case of non-marked point patterns, all the moments of the probability distribution of these so-called multifractal sets do not scale equivalently and an entire spectrum of generalized fractal dimensions $D_{q}$ is required \cite{Grass832,Hent83,Pal87,tel89,Bor93,Per06,Seur09}. In practice, for $q\neq 1$, $D_{q}$ is generally obtained with Rényi's information, $RI_{q}(\delta)$, of $q^{th}$ order (or Rényi's generalized entropy, \cite{Renyi70}) through a generalization of the box-counting method \cite{Hent83,Grass832,Pal87}:
\begin{equation} 
 RI_{q}(\delta) = \frac{1}{1-q} \ \log(\sum_{i=1}^{n_{box}(\delta)} p_{i}(\delta)^q)
\end{equation}
where $p_{i}(\delta) = n_i/N$ is the value of the probability mass function in the $i^{th}$ box of size $\delta$ and $q \in \mathbb{R}\backslash\lbrace-1\rbrace$. 

Then, for a multifractal point pattern, $\exp(RI_{q}(\delta))$ follows a power law:
\begin{equation} 
 \exp(RI_{q}(\delta)) \propto \delta^{-D_{q}}
\end{equation} and therefore   
\begin{equation}
 \label{eq:piren}\sum_{i=1}^{n_{box}(\delta)} p_{i}(\delta)^q \propto \delta^{(q-1)D_{q}}
\end{equation} where 
\begin{equation} 
 D_{q} = \lim_{\delta \to 0}\frac{RI_{q}(\delta)}{\log(\frac{1}{\delta})}
\end{equation}
For monofractal sets, $D_q$ is equal to $df_{box}$ for any order $q$, whereas, in the case of multifractal sets, $D_q$ decreases as $q$ increases (see Figure \ref{TheoSier}).\\

Finally, as complementary information, it can be noticed that: 
\begin{align}
 D_0 &= df_{box}\\
 \lim_{q \to 1}D_q &= df_{i}\\
 D_2 &= df_{sand}
\end{align}
where $df_{i}$ and $df_{sand}$ are, respectively, the information dimension \cite{Renyi56,Hent83,Seur09} and the dimension computed with the sandbox-counting method (also called the correlation dimension) \cite{Grass831,Fed88,tel89,TuiaKan08}. 

\subsection{The Lacunarity Index} 
Several monofractal sets can share the same fractal dimension $df_{box}$ and still be very different. The lacunarity index, which assesses the degree of translational invariance of a pattern, makes it possible to distinguish them. It is computed following the gliding box algorithm \cite{AlCl91,Plo96,TuiaKan08}: a box of size $\delta$ (diagonal) is superimposed on the origin of the set and the number of points $n_1$ falling into it is recorded; then, the box is moved one space further so as to partially overlap the previous location and $n_2$ is calculated. When the entire set has been covered, $\delta$ is increased and the algorithm is iterated. It must also be specified that each displacement of the box must be of the same length all along the computation. Finally, for a given $\delta$, the lacunarity index $\Lambda(\delta)$ is given by the following formula: \begin{equation}
 \label{eq:lam}\Lambda(\delta) = \frac{Z_{2}(\delta)}{Z_{1}(\delta)^2}
\end{equation}
where $Z_{q}(\delta)$ is the statistical moment of order $q$ of the probability function $P(n,\delta)$ which is equal to the number $n_{box}(n,\delta)$ of boxes of size $\delta$ containing $n$ points divided by the total number of boxes $Q$: 
\begin{equation} 
Z_{q}(\delta) = \sum_{n} n^q \frac{n_{box}(n,\delta)}{Q}  = \sum_{n} n^q P(n,\delta) = Q^{-1} \sum_{i=1}^{Q} n_i^q 
\end{equation}
Thus, if a set is highly lacunar (i.e. the mass distribution is characterized by a high variability), the lacunarity index is higher than 1. Now, for a fractal or multifractal set and for a grid made of $E$-cubes with $E$ being the dimension of the embedding space\cite{Grass832,Hent83,AlCl91}: 
\begin{align} 
 Z_{q}(\delta) &= \sum_{n} n^q \frac{n_{box}(n,\delta)}{Q}\\
 &= \frac{N^q}{Q}\sum_{i=1}^{Q} p_{i}(\delta)^q 
\end{align} which can be rewritten as follows using $\delta = \sqrt{E}\cdot\ell$ and $Q = \left(\frac{L}{\ell} \right)^E$ with $\ell$ being the length of a box edge (see Figure \ref{ComMori}):
\begin{align}
 Z_{q}(\delta) &= \frac{N^q}{Q \ \delta^E} \ \delta^E  \sum_{i=1}^{Q} p_{i}(\delta)^q \\
 &= \label{eq:proj}\frac{N^q}{L^E  \sqrt{E}^E} \ \delta^E  \sum_{i=1}^{Q} p_{i}(\delta)^q
\end{align}
Finally, $\frac{N^q}{L^E  \sqrt{E}^E}$ is a constant with regard to $\delta$. Notice that a similar conclusion can be reached for a grid made of $E$-orthotopes, but the development would involve the projection of $\delta$ onto a box edge. In practice, in order to avoid such a drawback, it is judicious to rescale the studied variables in $[0,1]$ for $E \geq 3$.\\

From equations \ref{eq:piren} and \ref{eq:proj}, the following power law can be deduced for fractal and multifractal sets: 
\begin{equation} 
 \label{eq:zp} Z_{q}(\delta) \propto \delta^{(q-1) D_{q}+E}
\end{equation} 
And, finally, equations \ref{eq:lam} and \ref{eq:zp} lead to \cite{AlCl91}: 
\begin{equation} 
 \label{eq:mindex_lac} \Lambda(\delta) = \nu \ \delta^{D_2-E}= \nu \ \delta^{-C_2}
\end{equation} 
where $C_2 = \vert D_2-E \vert$ is the co-dimension of order 2 and $\nu$ is the lacunarity parameter which entirely characterizes the lacunarity of self similar sets \cite{AlCl91}.\\

\subsection{The $m$-Morisita Index and Rényi's Generalized Dimensions}\label{linkMindex} 
It can be shown that the Morisita index $I_\delta \approx \Lambda(\delta)$ \cite{Dal02} for any pattern as long as $N \gg Q \geqslant 1$. Consequently, in the case of (multi)fractal sets, equation \ref{eq:mindex_lac} leads to the following power law: 
\begin{equation} 
 I_\delta \propto \delta^{-C_2} 
\end{equation} and 
\begin{equation} 
 \lim_{\delta \to \infty}\frac{\log{(I_\delta)}}{\log (\delta)}\approx-C_2
\end{equation} 
Now, this relationship can be extended to the $m$-Morisita index for $ m \in\lbrace 2,3,4,\dotsm\rbrace$.
\begin{remark}
For a pattern $P$ with $N$ points, with a grid made of $Q$ cells and with $H:=\underset{i}{max}\left(n_i\right)$ such that $H\gg m$, $\frac{Z_j(\delta)}{Z_m(\delta)}$ is close to $0$ \ $\forall j \in\lbrace1,2,...,m-1\rbrace$. This follows from the fact that $H\gg m \Rightarrow Z_m(\delta) = Q^{-1}\sum_{i=1}^{Q} n_i^m\gg Q^{-1}\sum_{i=1}^{Q} n_i^{m-1}=Z_{m-1}(\delta)$. 
\end{remark}
\begin{prop}
Let $P$ be a pattern with $N$ points and with a grid made of $Q$ cells such that $H\gg m$, then: $I_{m,\delta}\approx\frac{Z_m(\delta)}{Z_1(\delta)^m}$ \ $\forall m \in \lbrace2,3,\dotsm \rbrace$.
\end{prop}
\begin{proof}
\begin{align} 
 I_{m,\delta}&=Q^{m-1}\ \frac{\sum_{i=1}^Q n_i(n_i-1) \dotsm (n_i-m+1)}{N(N-1) \dotsm (N-m+1)}\\
 &=Q^{m-1} \ \frac{\sum_{i=1}^Q n_i^m + \alpha_{m-1}\sum_{i=1}^Q n_i^{m-1} + \dotsm +\alpha_1\sum_{i=1}^Q n_i}{N^m + \alpha_{m-1}N^{m-1} + \dotsm + \alpha_1 N}\\ 
 &=Q^{m-1} \ \frac{Q Z_m(\delta) + \alpha_{m-1} Q Z_{m-1}(\delta) + \dotsm + \alpha_1 Q Z_1(\delta)}{N^m\left[ 1 + \alpha_{m-1} \frac{1}{N} + \dotsm + \alpha_1 \frac{1}{N^{m-1}}\right]}\\
 &=\frac{Q^m Z_m(\delta)}{N^m} \ \frac{1 + \alpha_{m-1} \frac{Z_{m-1}(\delta)}{Z_m(\delta)} + \dotsm + \alpha_1 \frac{Z_1(\delta)}{Z_m(\delta)}}{1 + \alpha_{m-1} \frac{1}{N} + \dotsm + \alpha_1 \frac{1}{N^{m-1}}}
\end{align}
Because $\frac{Z_j(\delta)}{Z_m(\delta)}$ and $\frac{1}{N^j}$ are close to $0$ \ $\forall j \in\lbrace1,2,...,m-1\rbrace$ (respectively due to the Remark and to the fact that $N \gg m$), then: 
\begin{equation} 
 I_{m,\delta}\approx\frac{Q^m Z_m(\delta)}{N^m}=\frac{Q^m Z_m(\delta)}{(Q Z_1(\delta))^m}=\frac{Z_m(\delta)}{Z_1(\delta)^m}
\end{equation} 
\hskip10cm$\blacksquare$ 
\end{proof}
For (multi)fractal sets and for orders $q=m$, equations 18 and 26 lead to the following power law: 
\begin{equation} 
 I_{m,\delta} \propto \delta^{(m-1)(D_{m}-E)}
\end{equation} 
and 
\begin{align}
 \label{MI_Dq} \lim_{\delta \to \infty}\frac{\log{(I_{m,\delta})}}{\log (\delta)}&\approx(m-1)(D_{m}-E)\\
  &\approx -(m-1)(C_{m})\\
  &\approx -S_m
\end{align}
where $C_{m}$ is the co-dimension of order $q=m$, $S_m$ will be called the $m$-Morisita slope and the dependence between $S_m$ and $m$ will be referred to as the $m$-Morisita slope spectrum.\\

\begin{figure}[t]
\begin{center}
\includegraphics[width=13.5cm]{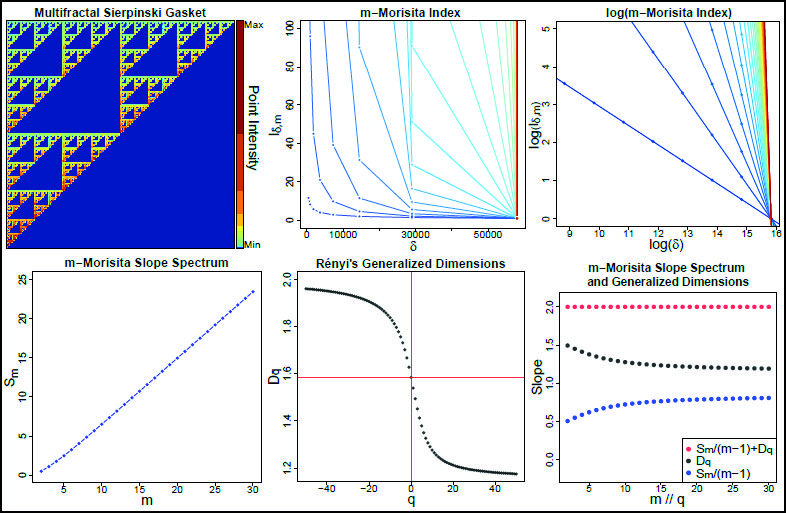}
\end{center}
\vspace{-20pt}
\caption{Illustration of the theory on a multifractal Sierpinski gasket for $m\in \lbrace 2,3,\cdots,30\rbrace$}
\label{TheoSier}
\end{figure}
 
In practice, $-S_m$ can be estimated as the slope of the linear regression fitting the data of the plot relating $\log$($I_{m,\delta}$) to $\log$($\delta$). Besides, as for $df_{box}$, $S_m$ can be used to assess the degree of clustering of point patterns. It ranges between $0$ (for regular patterns) and $(m-1)E$ (for highly clustered patterns). Again, point patterns are rarely (multi)fractal and $S_m$ must be interpreted as the slope of the (multi)fractal pattern most similar to the studied one. Moreover, although when working with spatially clustered data the condition $H\gg m$ is easily met, it is not the case when working with regular or random patterns and such distributions must be handled with care. Consequently, it is good practice to always check how well a straight line fit the data in the $\log$-$\log$ plot relating $\log$($I_{m,\delta}$) to $\log$($\delta$).\\

Finally, a multifractal Sierpinski gasket was considered (see Figure \ref{TheoSier}) and the different measures presented in this section were applied on this theoretical point set. The bottom-right panel illustrates the relationship of equation \ref{MI_Dq}.

\section{Clustering Characterization of Environmental Monitoring Networks}\label{cluEMN}

In this section, a complete methodology for the analysis and quantification of clustering intensity with the $m$-Morisita index is presented. A good methodology must be able to take into account both natural (e.g. shapes of lakes and forests) and administrative constraints (e.g. the finiteness and irregularities induced by administrative borders) which might be irrelevant to a monitored phenomenon. These considerations are important because an ideal EMN filling randomly or dispersedly a study region delimited by tortuous borders would appear clustered to most measures of clustering, although the representativeness of the collected data could be good (i.e. could not be improved with declustering algorithms).

\begin{figure}
\begin{center}
\includegraphics[width=13.5cm]{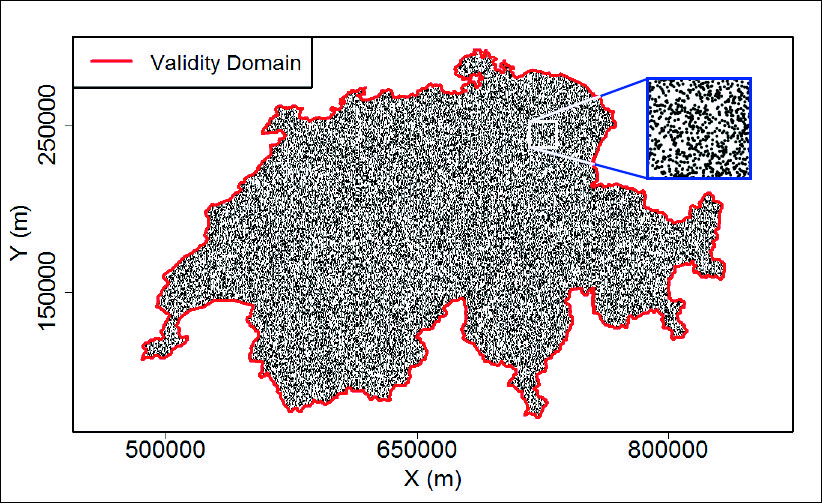}
\end{center}
\vspace{-20pt}
\caption{The retained validity domain and one of the random point distribution used for comparison}
\label{VD}
\end{figure}

\subsection{The Proposed Methodology} 
The present methodology involves Monte-Carlo simulations along with the concept of validity domain and avoids resorting to complicated edge effect corrections \cite{Ripl81}. It is subdivided into four steps \cite{KanMai04}:
\begin{enumerate}
 \item Selection of a dataset provided by an EMN composed of $s$ measurement sites.
 \item Selection of a validity domain (i.e. space of interest).
 \item Generation of many random patterns within the validity domain (i.e. Monte Carlo simulations) generated from a uniform distribution. Each simulation must be composed of $s$ points. In this way, a reference level of clustering along with a confidence level can be later obtained \cite{Ill08}.
 \item The $m$-Morisita index (Equation \ref{eq:mMorisita}) is applied to both the raw and simulated data and the different results are analysed and compared. If needed, statistical tests can be conducted \cite{BesDig77}.    
\end{enumerate}

Here, as mentioned in the introduction, the challenging case study of the SIRMN was used (see Figure \ref{RadonEMN}) and only inhabited dwellings and ground floor levels were taken into account. The validity domain was delimited using the administrative borders of Switzerland. One hundred random point distributions were simulated within the limit of the validity domain (see Figure \ref{VD}) and, finally, the $m$-Morisita index was applied to both raw and simulated patterns for $m \in \lbrace 2,3,4,5\rbrace$.

\begin{figure}
\begin{center}
\includegraphics[width=13.5cm]{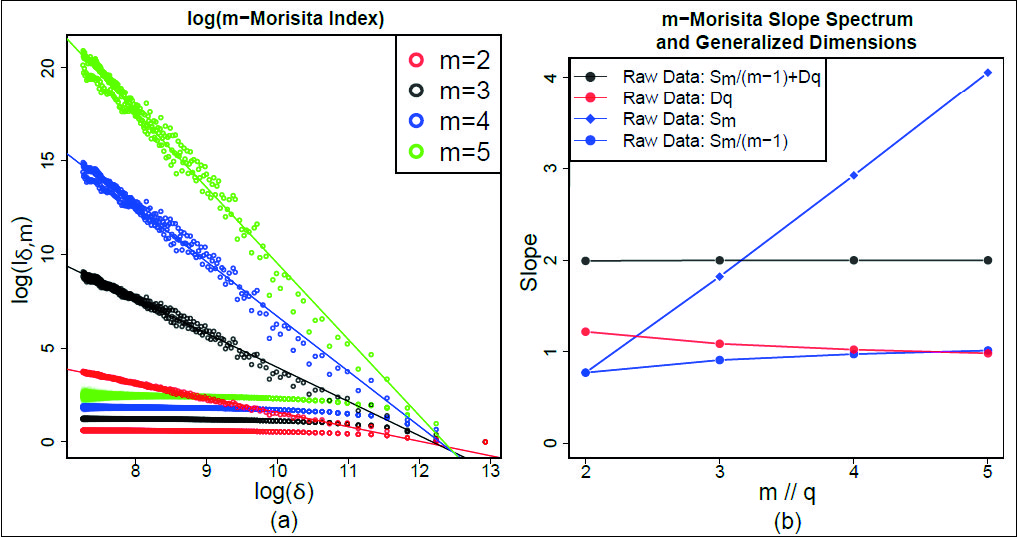}
\end{center}
\vspace{-20pt}
\caption{(a) SIRMN characterization by the $m$-Morisita index. (b) Link to Rényi's generalized dimensions. The "log" refers to "natural logarithm" and $\delta$ is measured in meters.}
\label{res1}
\end{figure}

\subsection{Results}
The results are displayed in Figure \ref{res1}. For the raw pattern, a fan of four point clouds corresponding to each $m$ is obtained in the plot relating $\log$($I_{m,\delta}$) to $\log(\delta)$ (see Figure \ref{res1} (a)). Their general behaviour can reasonably be approximated by four lines whose slopes respect the relationship given by equation \ref{MI_Dq} (see Figure \ref{res1} (b)). As $m$ increases, the $m$-Morisita index becomes more and more sensitive to the distinctive features of the pattern. This is highlighted by the increase in $S_m$ values of the $m$-Morisita slope spectrum. Next, the point clouds provided for the 100 simulated patterns are displayed all together in the bottom part of the $\log$-$\log$ plot. For each $m$, as $\log(\delta)$ decreases, a quick evolution to a steady state is observed, which is coherent with random patterns covering entirely a study area delimited by a validity domain. As $m$ increases, it also becomes easier to distinguish between the behaviour of the raw and simulated patterns. Besides, since the results of the raw pattern don't fall into the distribution sketched by the simulations, there is no need to resort to statistical tests to claim for the statistical significance of the observed differences. Consequently, a declustering algorithm should be performed before moving to the modelling phase in order to avoid local overestimations or underestimations of radon concentrations.\\

Finally, the use of the $m$-Morisita index is particularly relevant in cases where $m>2$ is required to detect differences. The next subsection shows such an example using simulated data. 

\begin{figure}[!ht]
\begin{center}
\includegraphics[width=13.5cm]{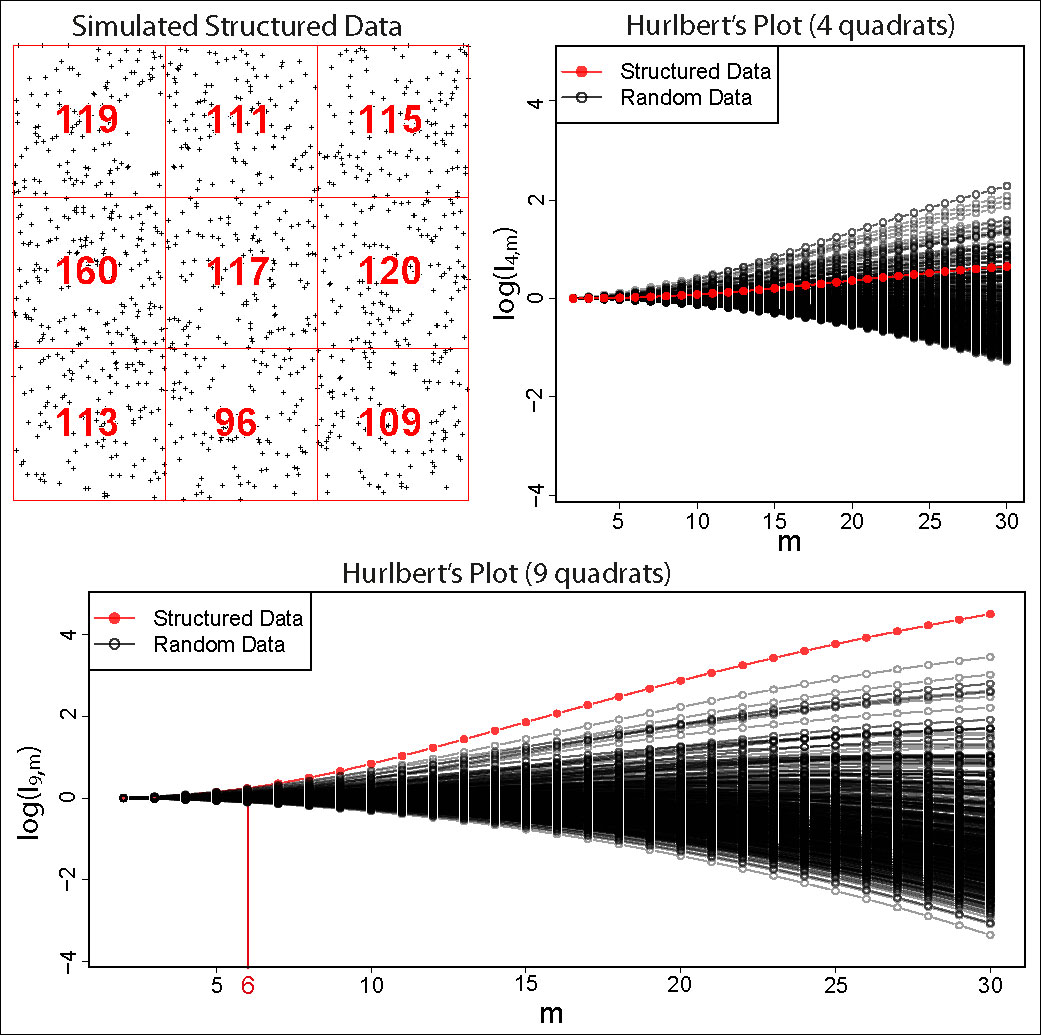}
\end{center}
\vspace{-20pt}
\caption{Simulated data with a known structure and Hurlbert's corresponding plots for two different scales expressed in number of quadrats.}
\label{HurlbertPlot}
\end{figure}

\subsection{Further Considerations Based on Simulated Data}
In this subsection, an application of the $m$-Morisita index to simulated data with a known structure is proposed. The purpose is to emphasize the importance of using both increasing $m$ and different $\delta$ to capture departures from random situations.\\

The simulated point pattern with a known structure used in this subsection is displayed in the top-left plot of Figure \ref{HurlbertPlot}. It was generated as follows using an R package called Spatstat \cite{BaddeTurner05}:
\begin{enumerate}
\item 1000 points were generated from a uniform distribution within a square.
\item A grid of nine quadrats was overlaid over the pattern (in red in the top-left plot of Figure \ref{HurlbertPlot})
\item 60 additional points were generated from a uniform distribution inside one of the nine quadrats (randomly selected). In the top-left plot of Figure \ref{HurlbertPlot}, the selected quadrat happened to be the middle one of the left column with 160 points.
\end{enumerate}

$I_{m,\delta}$ was computed on the resulting pattern for $m\in \lbrace 2,3,\cdots,30\rbrace$ and for two different $\delta$ corresponding to two grids made of four and nine quadrats respectively. The dependence between $I_{m,\delta}$ and $m$ was recorded separately for each $\delta$. The same was done with 500 random patterns (uniform distribution) made of 1060 points and the results are displayed in Hurlbert's plots \cite{Hul90} of Figure \ref{HurlbertPlot}. At the scale corresponding to the grid of four quadrats (see Figure \ref{HurlbertPlot} top-right), the behaviour of the structured pattern does not depart from that of the random patterns (i.e. the red curve is completely included into the envelope sketched by the black ones). In contrast, for the grid of nine quadrats, the red curve of the structured pattern extricates itself from the set of black curves, but only for $m\geqslant 6$.\\

Finally, this example reinforces the importance of $I_{m,\delta}$ such as it is defined in Equation \ref{eq:mMorisita}. With the SIRMN study case, it was already shown that the sensitivity of the index increased with $m$. Here, in addition, it is highlighted that this increasing sensitivity, along with the concept of scale, is of paramount importance. Indeed, the known structure of the simulated pattern is clearly detected only for $m \geqslant 6$ and for only one of the two tested scales.

\section{Structure Detection Using Functional Measures of Clustering}\label{StrDet}
The purpose of this section is to introduce a simple methodology using the $m$-Morisita index to detect structures in monitored phenomena when traditional geostatistical tools are hard to implement. The fundamental idea is to compare the spatial clustering of reference random patterns (produced by shuffling the original one) with the raw data clustering (i.e. clustering of the measurement sites) at different levels of the measured function (i.e. radon concentration) by applying different thresholds to the raw data, i.e. by performing FMC \cite{lov87}. Visually, the proposed methodology aims at quantifying the difference between the two maps displayed in Figure \ref{3D}.\\

\begin{figure}[!ht]
\begin{center}
\includegraphics[width=13.5cm]{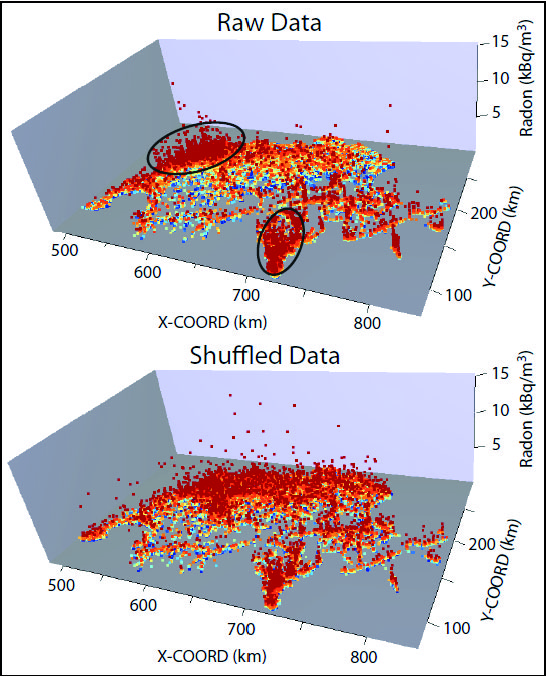}
\end{center}
\vspace{-20pt}
\caption{Indoor radon concentrations in Switzerland and shuffled data}
\label{3D}
\end{figure}

In this section, the specificity of the suggested methodology is explained step-by-step and the main results are analysed and discussed. The considered data are the same as those of the previous section (i.e. the SIRMN).

\subsection{The Proposed Methodology}           
The proposed methodology is based on the $m$-Morisita slope and FMC. It is subdivided into five steps:
\begin{enumerate}
\item Selection of a spatial dataset provided by an EMN (i.e. the raw dataset).
\item Perform many shufflings of the variable of interest to generate the shuffled datasets. In details, it consists in separating the variable of interest from the location coordinates. Then, the values of the variable are shuffled before being put back to the coordinates. Each time the operation is iterated, a new shuffled dataset is produced.  
\item Quantiles (i.e. quintiles or deciles depending on the number of points) of the studied variable are used as thresholds to split up the raw and shuffled datasets. 
\item $S_m$ is used to estimate the degree of clustering of the raw and shuffled datasets above each threshold (i.e. application of the functional $m$-Morisita index to the raw and shuffled datasets). This step requires that condition $H\gg m$ be respected for each subset (see Subsection \ref{linkMindex}).
\item If needed, statistical tests of significance can be conducted based on the distribution sketched by the shuffled datasets.    
\end{enumerate}

In the SIRMN study case, the decile thresholds were applied; 100 hundred shuffled datasets were generated and the $m$-Morisita slope was used with $m\in\lbrace2,3,4,5\rbrace$.

\subsection{Results} 
From a general perspective, the results provided by the functional $m$-Morisita index (see Figure \ref{res2_3}) show that the clustering intensity increases with thresholds. This is in accordance with what is visible in Figure \ref{3D}: the highest values of radon concentrations are also the more clustered and they mainly accumulate within the two highlighted Swiss regions \begin{figure}[h]
\begin{center}
\includegraphics[width=13.5cm]{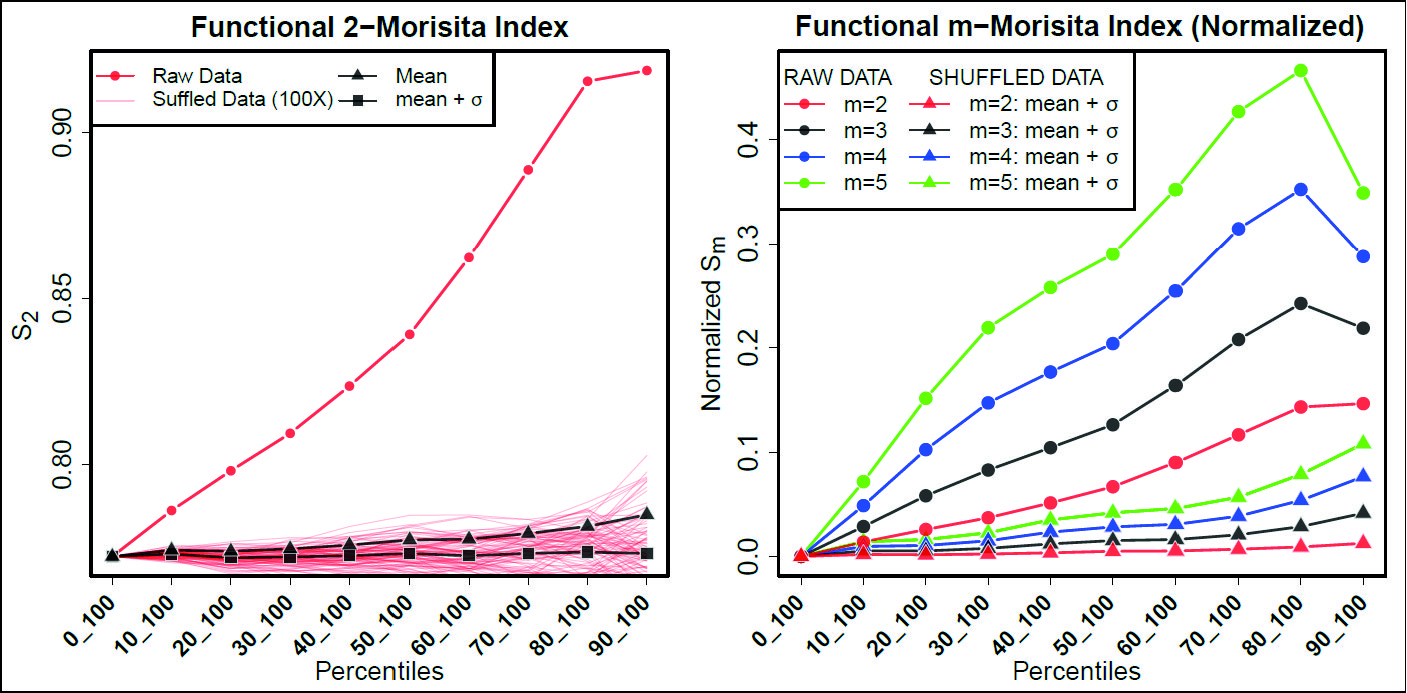}
\end{center}
\vspace{-20pt}
\caption{Results for $m$=2 (left) and result comparison for $m\in\lbrace2,3,4,5\rbrace$(right). $\sigma$ refers to the standard deviation of the distributions sketched by the shuffled datasets for each thresholds.
\label{res2_3}}
\end{figure}Moreover, whatever $m$, the detected structure is statistically significant, since the evolution observed for the raw data does not fall into the distribution sketched by the shuffled datasets. Again, as $m$ increases, the functional $m$-Morisita index becomes more and more sensitive to the underlying structure of the indoor radon distribution. The ergodic fluctuations between the different simulations are also better captured with bigger values of $m$. This can be easily noticed through the evolution of the standard deviations with thresholds (see Figure \ref{res2_3}-right). But, in spite of this observation, the distance between the raw data lines and those of the standard deviations increases with $m$ and, consequently, the efficiency of the index in terms of structure detection follows the same progression.   

\section{Conclusion}
The $m$-Morisita index, such as defined in this paper, is a new powerful tool for the analysis of spatial patterns. Its close relationship to Rényi's generalized dimensions allows us to gain a deeper understanding of its behaviour when applied to complex point distributions. In practice, it is straightforward to use it through a simple methodology to characterize the degree of clustering of EMN measurement sites (i.e. the support of the measures). The results highlights the importance of considering multiple scales and shows that the sensitivity of the index increases with $m$. Based on its relationship to multifractality, the $m$-Morisita index can also be adapted to structure detection in monitored phenomena (i.e. the measures) through a second methodology based on FMC. From this perspective as well, the results benefit from the above-mentioned assets of the index (i.e. increasing sensitivity with $m$ and integration of multiple scales) and emphasize its high potential when applied to complex case studies. Besides, this second methodology is conceptually interesting, since it builds a bridge between the support of the measures and the measures themselves: the degree of clustering of the measurement sites is computed at different intensity thresholds of the monitored phenomena, which gives an insight into the spatial dependence of the measures.\\

In future research, the possibility of developing a new $m$-Morisita index for $m \in \mathbb{R}$ will be studied. The idea is to explore further the connection to Rényi's generalized dimensions. The use of multiple $m$ to extract information regarding the average size of clusters will be developed as well and new challenging case studies in high dimensional spaces will be considered. Finally, the influence of optimization methods on the results will be analysed. A special attention will be paid to methods involving information on both the support of the measures and the measures (e.g. methods based on conditional stochastic simulations in geostatistics and active learning methods using machine learning algorithms \cite{TuiaPoz13}).

\section{Acknowledgements}
The research was partly supported by the Swiss NSF project No. 200021-140658: "Analysis and modelling of space-time patterns in complex regions". The authors also want to thank the anonymous reviewers for their constructive comments and the Swiss federal office of public health for providing the radon data.

\bibliographystyle{elsarticle-num}
\bibliography{References}
\end{document}